\documentclass[a4paper,twocolumn,11pt, preprint]{quantumarticle}
\pdfoutput=1
\usepackage[utf8]{inputenc}
\usepackage[english]{babel}
\usepackage[T1]{fontenc}
\usepackage{amsmath}

\usepackage{tikz}
\usepackage{lipsum}

\usepackage{subcaption}

\usepackage{amsthm}
\usepackage{amssymb}
\usepackage{mathtools}
\usepackage{graphicx}
\usepackage{braket}
\usepackage[colorlinks=true, allcolors=blue]{hyperref}

\usepackage{bbm}
\usepackage{comment}

\usepackage[most]{tcolorbox}

\newtheorem{proposition}{Proposition}

\renewcommand{\k}{\boldsymbol{k}}
\newcommand{\x}{\boldsymbol{x}}
\newcommand{\y}{\boldsymbol{y}}

\newcommand{\z}{\boldsymbol{z}}
\newcommand{\s}{\boldsymbol{s}}

\title{Efficient training of photonic quantum generative models}
\author{Felix Gottlieb}
\affiliation{Quandela, 7 Rue Léonard de Vinci, 91300 Massy, France}
\author{Chayma Faraji}
\affiliation{Quandela, 7 Rue Léonard de Vinci, 91300 Massy, France}
\author{Rawad Mezher}
\affiliation{Quandela, 7 Rue Léonard de Vinci, 91300 Massy, France}
\author{Brian Ventura}
\affiliation{Quandela, 7 Rue Léonard de Vinci, 91300 Massy, France}
\author{Shane Mansfield}
\affiliation{Quandela, 7 Rue Léonard de Vinci, 91300 Massy, France}
\author{Alexia Salavrakos}
\email{alexia.salavrakos@quandela.com}
\affiliation{Quandela, 7 Rue Léonard de Vinci, 91300 Massy, France}

\begin{document}

\maketitle

\begin{abstract}
The topic of generative learning has gained traction within the field of quantum machine learning, in particular with the advent of train-on-classical, deploy-on-quantum methods. This approach exploits the properties of intermediate-complexity circuits whose training can be simulated classically efficiently, but that generally require quantum hardware for the corresponding sampling problem. Quantum linear optics possess similar properties, which allow us to propose an efficient training procedure for photon-native quantum generative models based on the maximum mean discrepancy, where the deployment of the model corresponds to the task of boson sampling. We provide numerical results, propose datasets, and we also explore how initialization strategies, kernel and ansatz choice affect the training.
\end{abstract}

\section{Introduction}
The field of quantum machine learning (QML) is faced with two main challenges: training quantum models efficiently and at scale; and identifying practically useful problems for which quantum models outperform classical ones -- ideally, showing the model is non-dequantizable. As highlighted in \cite{gilfuster2025relationtrainabilitydequantizationvariational}, no work to date fulfills all criteria. However, noteworthy results in this direction include the first example of a rigorous speedup in QML \cite{Liu_2021}, the guarantee of backpropagation scaling for a specific family of circuits \cite{Bowles_2025}, as well as frameworks for formal learning separations \cite{gyurik2024exponentialseparationsclassicalquantum, molteni2025quantummachinelearningadvantages}. 

In the context of generative learning, the recent work of \cite{Recio-Armengol2025} proposes a compelling approach that targets the problem of training QML models efficiently and at scale. The authors present a generative model based on parametrized Instantaneous Quantum Polynomial (IQP) circuits \cite{Shepherd_2009}. Two facts about these circuits combine to make training classically efficient. First, for IQP circuits the expectation value of observables can be approximated  using den Nest's classical algorithm \cite{nest2010simulating}. Second, in the related work of \cite{rudolph2023trainability}, a loss function used in generative learning known as the Maximum Mean Discrepancy (MMD) \cite{Gretton_2012} is shown to correspond to expectation values of circuit observables. Since expectation values are classically efficient for IQP circuits, this loss function can be evaluated classically.

%for which it is known that expectation values of observables can be approximated classically efficiently through the den Nest algorithm \cite{nest2010simulating}. Importantly, in the related work of \cite{rudolph2023trainability}, the authors consider a loss function often used in generative learning called the Maximum Mean Discrepancy (MMD) \cite{Gretton_2012} and obtain a way to express it as a quantum observable. This means that by using this MMD observable as a loss function, IQP-based models can be trained classically efficiently.

Thus, successful tools from classical ML like automatic differentiation and backpropagation can be used in the training of the model. This solves a critical challenge in the optimization of parametrized quantum circuits: it allows us to bypass estimating gradients on quantum hardware via the parameter-shift rule or other methods whose cost scales with the number of parameters in the circuit, and instead inherit the favorable scaling of backpropagation.

%indeed, without backpropagation scaling, each estimation of the gradient of the loss function requires as many circuit evaluations as there are parameters, making gradient descent very costly as circuits grow in size.

At the same time, in generative learning, the task of deploying the model simply means generating new data samples -- in the case of parametrized quantum circuits, this only requires sampling from the circuit through measurement. Crucially, the task of sampling from IQP circuits remains hard classically in general, despite the den Nest algorithm. This means that quantum hardware would likely still be required for the model to be deployed and used. 

Interestingly, IQP circuits are not the only circuits with these properties. In particular, discrete variable linear optical quantum computing (DVLOQC) also offers this level of ``intermediate complexity''. Indeed, boson sampling is expected to be hard classically \cite{aaronson2011}, while there exist efficient classical algorithms for estimating the expectation value of observables in linear optics \cite{Gurvits_2005, lim2025}. We thus set out to extend this framework to DVLOQC. 

% Paragraph with more details
In this work, we propose a training procedure for photon-native \cite{Salavrakos_2025_perspective} quantum generative models based on the MMD and on Gurvits' algorithm for the approximation of expectation values of observables in linear optics. We train our model on datasets that can be represented as fixed Hamming-weight bitstrings, as this corresponds to the natural outputs of a linear optical circuit. In our simulations we train models with up to 16 photons in 256 modes and more than 100 000 parameters in reasonable time on a laptop. We explore initialization strategies as well as different interferometer meshes, kernel choices, and parametrizations of the ansatz. Our code framework is available on the MerLin platform \cite{notton2026merlindiscoveryenginephotonic}, Quandela's open-source Python framework for photonic QML.

\section{Preliminaries}

\subsection{Discrete variable linear optical quantum computing}\label{sec:dvloqc} We consider a setup where a source emits single photons which are prepared as Fock states $\ket{s} \in \Phi_{m, n}$, where 
\begin{equation}
    \Phi_{m, n} = \{\ket{s_1, \dots, s_m} \big|  \sum_i s_i = n,  s_i \in \{0, \dots n\}\}, \nonumber
\end{equation} 
with $m$ the number of optical modes and $n$ the number of photons. State $\ket{s}$ is sent as an input to a circuit which is a universal interferometer made of beam-splitters and phase-shifters. Such an interferometer is said to be universal as it can express any $m \times m$ unitary $U$, i.e. any unitary evolution in the single-particle space where state $\ket{s}$ contains only $n = 1$ photon. When $n \geq 1$, the interferometer applies the operator $\Tilde{U} = \varphi(U)$, which is the extension of $U$ on the multi-particle space\footnote{See, e.g. \cite{aaronson2011} for a description of mapping $\varphi$. Note that we adopt the tilde notation for operators acting on the multi-particle space.}.

At the output of the circuit, there are photon number detectors measuring in the Fock basis of states in $\Phi_{m, n}$, thus producing output arrangements of the form $\s' = (s'_1, \dots s'_m)$, with $\sum_i s'_i = n$ in the lossless case where the number of photons $n$ is conserved. For an input state $\ket{s^{in}}$ and an output state $\ket{s^{out}}$, the corresponding probability amplitude is given by
\begin{equation}
\bra{s^{out}} \Tilde{U} \ket{s^{in}} = \frac{\text{Perm}(U_{s^{in}, s^{out}})}{\sqrt{\Pi_i s^{in}_i! \Pi_j s^{out}_j!}},
\end{equation}
where $U_{s^{in}, s^{out}}$ is a submatrix of the single-particle unitary $U$, formed by taking $s^{out}_i$ times the $i$th row of $U$, then $s^{in}_j$ times the $j$th column of that intermediate matrix. 

Sampling from the circuit that prepares the state $\Tilde{U}\ket{\psi}$ is a task known as \emph{boson sampling}. Both exact sampling and approximate sampling are expected to be classically hard -- provided some conditions on $U$, as conjectured by Aaronson and Arkhipov \cite{aaronson2011}. 

For estimation problems however, the situation is different. In general, an observable on the Fock space is defined as a normal operator $\mathcal{Q}$ that satisfies $\mathcal{Q} = \Tilde{U}^\dagger \Lambda \Tilde{U}$, where $\Lambda$ is diagonal in the Fock basis: $\Lambda \ket{s} = \lambda(s) \ket{s}$ \cite{defelice2024}. The expectation value of $\mathcal{Q}$ is thus:
\begin{align}
E (\mathcal{Q}) = \bra{\psi} \Tilde{U}^\dagger \Lambda \Tilde{U} \ket{\psi} & = \sum_{s \in \Phi_{m, n}} \lambda(s) P_U (s | \psi) \nonumber \\
& = \sum_{s \in \Phi_{m, n}} \lambda(s) \|\bra{s} \Tilde{U} \ket{\psi}\|^2 \nonumber \\
& = \sum_{s \in \Phi_{m, n}} \lambda(s) \left| \frac{\text{Perm}(U_{\psi, s})}{\sqrt{\Pi_i s_i! \Pi_j \psi_j!}} \right|^2 \nonumber
\end{align}
for some input state $\ket{\psi}$. Note that $\mathcal{Q}$ is defined to be normal rather than Hermitian which means that the eigenfunction $\lambda$ can be complex valued.

Now, if we start from an operator defined on the single-particle space $Q = U^{\dagger} D U$ where $U$ is unitary and $D = \text{diag}(l_1, \dots l_m)$, then we can map it to an observable $\Tilde{Q}$ on the multi-particle space. There, its eigenvalues are given by the products of the $l_i$ via the following eigenfunction: 
\begin{equation}
\lambda(s) = \prod_i l_i^{s_i},
\label{eq:eigenvalues_lambda}
\end{equation}
where the $s_i$ are the mode occupancy numbers of Fock state $\ket{s}$. We then have the following bound on the singular value of $\Tilde{Q}$: $|\lambda(s)| \leq \text{max} (|l_i|)^n = ||Q||_2^n$. Additionally, the expectation value of such an observable can be expressed by a single permanent:
\begin{equation}
E (\Tilde{Q}) = \bra{\psi} \Tilde{Q} \ket{\psi} = \frac{\text{Perm}(Q_{\psi, \psi})}{\sqrt{\Pi_i \psi_i! \Pi_j \psi_j!}}. \label{eq:expectation_q}
\end{equation}
Importantly, this quantity can be approximated classically in $O(n^2/\epsilon^2)$ with precision $\pm \epsilon \|Q\|_2^n$ using Gurvits' algorithm and its generalizations \cite{Gurvits_2005, aaronson2012}\footnote{See also \cite{lim2025} for further results on the classical estimation of probabilities and expectation values in DVLOQC, including input Fock states in superposition.}. Hence, if $\|Q\|_2 \leq 1$, the quantity (\ref{eq:expectation_q}) can be approximated classically \emph{efficiently}. Note that Gurvits' algorithm is based on the Glynn estimator, which for a given $x\in\left\{
-1,1\right\}^{n}$ is defined for an $n\times n$
matrix $A=(a_{i,j})$ as
\begin{equation}
\operatorname*{Gly}\nolimits_{x}(A):=x_{1}\dotsm x_{n}\prod_{i=1}^{n}\left(
a_{i,1}x_{1}+\dotsb+a_{i,n}x_{n}\right).
\label{eq:glynn}
\end{equation}
More precisely, we have that:
\begin{equation}
\operatorname*{Per}\left(  A\right)  = \mathbb{E}_{x\in\{-1,1\}^{n}%
}\left[  {\operatorname*{Gly}\nolimits_{x}(A)}\right]. \nonumber
\end{equation}

\subsection{Generative model}\label{sec:generative_model}
Our model is the photonic quantum circuit Born machine (QCBM) described in \cite{Salavrakos_2025}, i.e. a variational quantum generative model based on DVLOQC with $n$ photons and $m$ modes. In a nutshell, the model consists of an input Fock state $\ket{s}$ that goes through a parametrized linear optical interferometer that prepares state $\ket{\psi_\theta} = \Tilde{U}_\theta \ket{s}$, which is then measured by the photon detectors. The outputs strings $\s' = (s'_i, \dots s'_m)$ are the generated data: each measurement output vector $\s'$ is considered to be a sample generated by the model.

The parameters $\theta$ of the model are the phases of phase shifters inside the circuit interferometer, which is thus described by an $m \times m$ matrix $U_\theta$. This matrix can be decomposed into different linear optical ans{\"a}tze: in the case of the rectangular interferometer \cite{Clements_2016}, $U_\theta$ is decomposed in a product of $T(\theta^{j}_{i}, \theta'^{j}_{i})$ matrices that each correspond to a Mach-Zender Interferometer (MZI) element between the  $j$-th and $(j+1)$-th modes:
\begin{equation*}
    \begin{pmatrix}
1&0&\cdots&&&\cdots&&0\\
0&1&&&&&&\vdots\\
\vdots&&&e^{i\theta^{j}_{i}} \cos{\frac{\theta'^{j}_{i}}{2}}&-\sin{\frac{\theta'^{j}_{i}}{2}}&\\
&&&e^{i\theta^{j}_{i}} \sin{\frac{\theta'^{j}_{i}}{2}}&\cos{\frac{\theta'^{j}_{i}}{2}}&&&\vdots\\
\vdots&&&&&&1&0\\
0&\cdots&&&&\cdots&0&1\\
\end{pmatrix}
\end{equation*}
Overall, the resulting $U_\theta$ matrix can be expressed as a product of $m(m-1)/2$ such matrices and a diagonal matrix. There is a total of $m^2$ phases in the interferometer, i.e. $m^2$ parameters in the model. Note that this is only one way of parametrizing $U_\theta$, as we discuss in Section \ref{sec:ansatze}.

As for the learning problem, a target dataset is given which contains $N$ training data samples which should also be of the form $\x = (x_1, \dots x_m)$. The goal is to train the model parameters $\theta$ so that the samples $\s'$ produced by the model resemble the target samples $\x$ according to a chosen loss function and metrics.

%\subsection{No-collision case}
%As shown in \cite{aaronson2011}, if the number of modes in a linear optical experiment is such that $m >> n^2$, the probability that there are two or more photons are in the same mode is not large and we can consider that we only have $1$ or $0$ photons per mode, especially in the limit of large $m$ and $n$. In this work, we focus on this regime both because it allows us to straightforwardly extend the results of \cite{rudolph2023trainability, Recio-Armengol2025}, but also because in practice in experimental setups, threshold detectors are more readily available than photon number resolving detectors. 

\section{Estimating the MMD loss in linear optics}
\subsection{MMD observable}\label{sec:MMD}
One such metric used in generative learning is the MMD, which was introduced in the context of classical ML in \cite{Gretton_2012}. It measures a distance between distributions $p$ and $q$ and is expressed as:
\begin{align}\label{eq:truemmd}
\text{MMD}^2(p, q) = \ &  \mathbb{E}_{\x,\y\sim p}\left[K(\x,\y)\right] \nonumber \\ & - 2\mathbb{E}_{\x\sim p, \y\sim q}\left[K(\x,\y)\right] \nonumber \\ & +  \mathbb{E}_{\x,\y\sim q}\left[K(\x,\y)\right], 
\end{align}
where $K(\x, \y)$ is a kernel function. We pick a common kernel choice studied in \cite{rudolph2023trainability, Recio-Armengol2025}, the Gaussian kernel with bandwidth parameter $\sigma$ :
\begin{align}\label{eq:kernel}
    K_\sigma (\x, \y) = \exp\left(\frac{-\vert\vert \x-\y \vert\vert^2}{2\sigma^2}\right),
\end{align}
where $\vert\vert \x-\y \vert\vert^2 = \sum_i (x_i - y_i)^2$. Note that another way to write the MMD is the following:
\begin{align}\label{eq:truemmd2}
\text{MMD}^2(p, q) = \ &  \sum_{\x, \y \in \mathcal{X}} p(\x) p(\y) K(\x, \y) \nonumber \\
\ & - 2 \sum_{\x, \y \in \mathcal{X}} p(\x) q(\y) K(\x, \y) \nonumber \\ 
\ & + \sum_{\x, \y \in \mathcal{X}} q(\x) q(\y) K(\x, \y).
\end{align}
An important property of the MMD is that it can be computed as an \emph{average over samples drawn from the target and model distributions}, thus making it an implicit loss function for a generative model, as pointed out in \cite{rudolph2023trainability}. Indeed, it can be obtained by sampling batches of vectors $\mathcal{X}=\{\x_i\sim p\}$ and $\mathcal{Y}=\{\y_j\sim q\}$ and using the following estimator of the $\text{MMD}^2$ which is an unbiased estimator: 
\begin{align}\label{eq:MMD_greton}
\hat{\text{MMD}}^2(\mathcal{X},\mathcal{Y}) = \ & \frac{1}{|\mathcal{X}|(|\mathcal{X}|-1)} \sum_{i\neq j} K(\x_i, \x_j) \nonumber \\ & - \frac{2}{|\mathcal{X}||\mathcal{Y}|}\sum_{i,j}K(\x_i,\y_j) \nonumber \\ & + \frac{1}{|\mathcal{Y}|(|\mathcal{Y}|-1)}\sum_{i\neq j}K(\y_i,\y_j).
\end{align} 

In \cite{rudolph2023trainability}, the authors highlight that, in QML, each of the three terms in the MMD loss function can be seen as the expectation value of an observable:
\begin{align}
\mathcal{M}(\rho, \rho') = \text{Tr}\left[ O_{\text{MMD}}^{(\sigma)} (\rho \otimes \rho')\right],
\end{align} 
with:
\begin{align}
O_{\text{MMD}}^{(\sigma)} = \sum_{\x, \y} K_{\sigma}(x, y) \ket{x}\bra{x} \otimes \ket{y}\bra{y},
\end{align} 
and where, depending on which term of the MMD is considered, $\rho$ and $\rho'$ can either be the state prepared by the quantum model, or a state representing the classical training data. Then, they show that $O_{\text{MMD}}^{(\sigma)}$ can be rewritten as a sum of Pauli strings $Z^{\k}$, which are tensor products of Pauli and identity operators on $m$ subsystems: 
\begin{equation}
Z^{\k} = Z_1^{k_1} \otimes \dots \otimes Z_m^{k_m}, \nonumber
\end{equation}
where $Z_i^{k_i} = Z_i$ if $k_i = 1$ and $Z_i^{k_i} = \mathbf{I}$ if $k_i = 0$. More precisely, they obtain: 
\begin{equation}
O_{\text{MMD}}^{(\sigma)} = \sum_{\k} (1-p_{\sigma})^{m-\vert \k \vert}p_{\sigma}^{\vert \k \vert} Z^{\k} \otimes Z^{\k},
\end{equation}
where $p_{\sigma}=\frac{1-e^{-\frac{1}{2\sigma^2}}}{2}$ and where $\vert \k \vert$ is the Hamming weight of $\k$. Putting everything together, this means that the MMD loss function amounts to a mixture of expectation values of observables:
\begin{align*}
%\label{eq:expmmdrudolph}
\text{MMD}^2(p,q) &= \mathbb{E}_{\k\sim \mathcal{P}_{\sigma}(\k)}\Big[ \big(\langle Z^{\k} \rangle_p  - \langle Z^{\k} \rangle_{q}\big)^2 \Big]
\end{align*}
where $\k$ is distributed according to a product of Bernoulli distributions with probability $p_{\sigma}$,
\begin{align}\label{eq:prob_p_sigma}
    \mathcal{P}_{\sigma}(\k) = (1-p_{\sigma})^{m-\vert \k \vert}p_{\sigma}^{\vert \k \vert}.
%\label{eq:kprobability}
\end{align}
Thus, the MMD loss function can be evaluated by sampling $\k$ vectors and computing the MMD as a sample average.

Now, let us return to linear optics. We define $W^{\k}$ as an operator in the single-particle space as the following diagonal $m \times m$ matrix:
\begin{equation}
    W^{\k} =
  \begin{pmatrix}
    (-1)^{k_1} & & \\
    & \ddots & \\
    & & (-1)^{k_m}
  \end{pmatrix},
\label{eq:z_mbym_matrix}
\end{equation}
and we denote $\Tilde{W}^{\k}$ its extension in the multi-particle space. Additionally, we define data domains $\mathcal{X}_{m, n}$ and $\Tilde{\mathcal{X}}_{m, n}$ of vectors of size $m$ with fixed Hamming-weight $n$:
\begin{align*}
& \mathcal{X}_{m, n} = \{(x_1, \dots, x_m) \big|  \sum_i x_i = n,  x_i \in \{0, \dots n\}\}, \\
& \Tilde{\mathcal{X}}_{m, n} = \{(x_1, \dots, x_m) \big|  \sum_i x_i = n,  x_i \in \{0, 1\}\}.
\end{align*}
This case of $\Tilde{\mathcal{X}}_{m, n}$ where the $x_i$ are limited to $0$ and $1$ values corresponds to the no-collision regime\footnote{This regime is quite relevant in practice in experimental setups, since threshold detectors are more readily available than photon number resolving detectors.} in linear optics, valid when $m >> n^2$. We can now state our main proposition.

%indeed, if the number of modes in a linear optical experiment is such that $m >> n^2$, the probability that there are two or more photons are in the same mode is not large and we can consider that we only have $1$ or $0$ photons per mode, especially in the limit of large $m$ and $n$. This regime is quite relevant in practice in experimental setups, since threshold detectors are more readily available than photon number resolving detectors. 

\begin{proposition}\label{proposition:main}
The MMD loss function for distributions $p, q$ defined on $\Tilde{\mathcal{X}}_{m, n}$ can be computed with the linear optical observable 
\begin{equation}
O_{\text{MMD, LO}}^{(\sigma)} = \sum_{\k} (1-p_{\sigma})^{m-\vert \k \vert}p_{\sigma}^{\vert \k \vert} \Tilde{W}^{\k} \otimes \Tilde{W}^{\k}
\nonumber
\end{equation}
in the no-collision regime.
\end{proposition}
\begin{proof}
We focus on the term of the MMD where both distributions come from the QCBM, as the classical terms are straightforward. Using expression (\ref{eq:truemmd2}), we wish to show that:
\begin{align}
\bra{\psi_\theta} O_{\text{MMD, LO}}^{(\sigma)} \ket{\psi_\theta} =  \sum_{\x, \y \in \Tilde{\mathcal{X}}_{m, n}} q(\x) q(\y) K_\sigma(\x, \y)
\label{eq:need_to_show}
\end{align}
Let us call this term $\text{MMD}_{q_\theta, q_\theta}$. Since the quantum circuit prepares state $\ket{\psi_\theta} = \Tilde{U}_\theta \ket{s}$, we have that, expanding $O_{\text{MMD, LO}}^{(\sigma)}$:
\begin{align*}
    \text{MMD}_{q_\theta, q_\theta} = \sum_{\k} & (1 - p_\sigma)^{m - |\k|} p_\sigma^{|\k|} \\ & . \bra{s} \Tilde{U}_\theta^\dagger \Tilde{W}^{\k} \Tilde{U}_\theta \ket{s} \bra{s} \Tilde{U}_\theta^\dagger \Tilde{W}^{\k} \Tilde{U}_\theta \ket{s}
\end{align*}
Here, 
\begin{align*}
\bra{s} \Tilde{U}_\theta^\dagger \Tilde{W}^{\k} \Tilde{U}_\theta \ket{s}  = \left( \sum_{x \in \Phi_{m, n}} \lambda^{\k}(x) P_{U_\theta}(x|s) \right),
\end{align*}
and, following equation (\ref{eq:eigenvalues_lambda}) and the form of $W^{\k}$:
\begin{align*} 
\lambda^{\k}(x) = \prod_{i = 1}^{m} \left( (-1)^{k_i}\right)^{x_i}.
\end{align*}
So, the expression becomes:
\begin{align*}
\text{MMD}_{q_\theta, q_\theta}
&= \sum_{x, y \in \mathcal{X}_{m, n}} q_{\theta}(x) q_{\theta}(y)  \\
&\quad . \sum_{\k} \Big((1 - p_\sigma)^{m - |\k|} p_\sigma^{|\k|}
      \prod_i (-1)^{k_i x_i + k_i y_i}\Big)
\end{align*}
%\begin{align*}
%    \text{MMD}_{q_\theta, q_\theta} = 
%    \sum_{x, y \in \mathcal{X}_{m, n}} q_{\theta}(x) q_{\theta}(y)  \sum_{\k} \Big((1 - p_\sigma)^{m - |\k|} p_\sigma^{|\k|} \\ .\prod_i (-1)^{k_i x_i + k_i y_i}\Big)
%\end{align*}
To recover (\ref{eq:need_to_show}), the sum over $\k$ in the above equation must correspond to the kernel $K_\sigma (\x, \y)$. At this stage, if we assume that we are in the no-collision regime and that $x_i, y_i \in \{0,1\}$, we can rewrite
\begin{align*}
\prod_i (-1)^{k_i(x_i + y_i)} & = \prod_i (-1)^{k_i(x_i - y_i)} \\
& = (-1)^{\sum_i k_i(x_i - y_i)}.
\end{align*}
Then the sum over $\k$ becomes
\begin{align*}
\sum_{\k} (1 - p_\sigma)^{m - |\k|} p_\sigma^{|\k|} (-1)^{\k(\x - \y)}
\end{align*}
which corresponds to the Fourier transform of the Gaussian kernel. More precisely, the Fourier transform is often called the Walsh-Hadamard transform when taken over spaces of bitstrings and is expressed as:
\begin{align*}
\hat{f}(\s)
= \frac{1}{2^m} \sum_{\x \in \{0,1\}^m} f(\x)\, (-1)^{\s \cdot \x}.
\end{align*}
Noting that $K_\sigma(\x,\y) = K_\sigma(\x - \y) = K_\sigma(\z)$, we have that:
\begin{align*}
\hat{f}(\k)
& = \frac{1}{2^m} \sum_{\z \in \{0,1\}^m} K_\sigma (\z)\, (-1)^{\k \cdot \z} \\
& = (1 - p_\sigma)^{m - |\k|} p_\sigma^{|\k|} 
\end{align*}
and thus, taking the inverse Fourier transform:
\begin{align*}
\sum_{\k \in \{0,1\}^m} (1 - p_\sigma)^{m - |\k|} p_\sigma^{|\k|} (-1)^{\k \cdot \z} = K_\sigma(\z).
\end{align*}
We then clearly recover:
\begin{align*}
    \text{MMD}_{q_\theta, q_\theta} = 
    \sum_{\x, \y \in \mathcal{X}_{m, n}} q_{\theta}(x) q_{\theta}(y) K_\sigma (\x, \y).
\end{align*}

\begin{comment}
Since the $k_i$ are independent, we can rewrite the sum over $\k$ as:
\begin{align*}
\prod_{i = 1}^{m} \sum_{k_i = 0}^{1} p_\sigma^{k_i} (1 - p_\sigma)^{1 - k_i}  (-1)^{k_i x_i + k_i y_i},
\end{align*}
where each factor in the product then equals:
\begin{align*}
(1 -  p_\sigma) + p_\sigma (-1)^{x_i + y_i}.
\end{align*}
We have two cases: if $(x_i + y_i) \text{ mod } 2$ is 0, each factor equals $1$, and if $(x_i + y_i) \text{ mod } 2$ is 1, each factor equals $1 - 2p_\sigma = e^{-1/2\sigma^2}$. Then, in the full product, the $e^{-1/2\sigma^2}$ will contribute a number of times equal to  $\sum_{i} (x_i + y_i) \text{ mod } 2$, and the rest of the factors will be equal to $1$, yielding:
\begin{align*}
    \exp\left(- \frac{1}{2\sigma^2} \sum_{i} ((x_i + y_i) \text{ mod } 2)\right).
\end{align*}

At this stage, if we assume that we are in the no-collision regime and that $x_i, y_i \in \{0, 1\}$, this is exactly: 
\begin{align*}
    \exp\left(- \frac{1}{2\sigma^2} \sum_{i = 1}^{m} (x_i - y_i)^2\right),
\end{align*}
which can also be rewritten in terms of the Hamming distance between $\x$ and $\y$:
\begin{align*}
    \exp\left(- \frac{1}{2\sigma^2} d_H(\x, \y)\right).
\end{align*}
This matches definition (\ref{eq:kernel}) of $K_\sigma (\x, \y)$, and thus:
\begin{align*}
    \text{MMD}_{q_\theta, q_\theta} = 
    \sum_{x, y \in \mathcal{X}_{m, n}} q_{\theta}(x) q_{\theta}(y) K_\sigma (\x, \y).
\end{align*}
\end{comment}
\end{proof}

Here, it is important to note that, since operator $W^{\k}$ is originally defined in the single-particle space, each expectation value $\bra{s} \Tilde{U}_\theta^\dagger \Tilde{W}^{\k} \Tilde{U}_\theta \ket{s}$ actually corresponds to a single permanent, as shown in equation (\ref{eq:expectation_q}). This means that the quantum linear optical terms of the MMD can be approximated classically efficiently with Gurvits' algorithm. We detail the exact procedure to follow in the next section.

\subsection{Classical estimation procedure}\label{sec:classical_getmmd}
In \cite{Recio-Armengol2025}, the authors define an unbiased MMD estimator for their classical estimation procedure. We can repurpose this estimator straightforwardly here, noting that the part which corresponds to the classical simulation of an IQP circuit through the den Nest algorithm can be replaced by the classical simulation of a linear optical circuit through the Gurvits algorithm. The method is presented in Figure \ref{fig:mmd_procedure}. %\textbf{Question: should we explicitely prove in an appendix that this remains an unbiased estimator?}

\begin{figure*}[t]
\centering
\begin{tcolorbox}
Sample the following batches: \begin{itemize}
        \item $\mathcal{X} = \{\x_i \sim p \}$ with $\x_i \in \{0, 1 \}^n$, which are the samples from the target dataset,
        \item $\mathcal{K} = \{\k_j \sim \mathcal{P} \}$ with $\k_j \in \{0, 1 \}^n$ to construct the MMD observable, where $\mathcal{P}$ corresponds the Fourier transform of the chosen kernel function for the MMD,
        \item $\mathcal{Z} = \{\z_l \sim U \}$ with $\z_l \in \{-1, 1 \}^n$ sampled uniformly to perform Gurvits' algorithm.
    \end{itemize}
Then, get unbiased estimates of the MMD$^2$ with the following formula: 
    \begin{align}
    \hat{\text{MMD}^2}(\mathcal{X},\mathcal{K}, \mathcal{Z}, \theta) & = \frac{1}{|\mathcal{K}| |\mathcal{Z}|(|\mathcal{Z}|-1)}\sum_{i,j,l\neq j}f(\k_i,\z_j, \theta)f(\k_i,\z_l, \theta) \nonumber \\ 
    &\quad\quad - \frac{2}{|\mathcal{K}| |\mathcal{Z}|  |\mathcal{X}|}\sum_{i,j,l}f(\k_i, \z_j, \theta)(-1)^{\x_l\cdot \k_i} \nonumber \\ 
    &\quad\quad\quad\quad + \frac{1}{|\mathcal{K}| |\mathcal{X}|( |\mathcal{X}|-1)}\sum_{i,j,l\neq j}(-1)^{\x_j\cdot \k_i}(-1)^{\x_l\cdot \k_i}, \label{eq:mmdestimate_us}
    \end{align}
    with the Glynn estimator
    \begin{align}
    f(\k, \z,\theta) & = \z_1 \cdots \z_n \prod_{i = 1}^{n} (q_{i, 1} \z_1 + \cdots + q_{i, n} \z_n),
    \end{align} 
    where the $q_{i, j}$ are the entries of matrix $Q_{\k, \theta}^{s, s}$, which is a submatrix of $Q_{\k, \theta} = U_\theta^\dagger W^{\k} U_\theta$  formed by taking $s_i$ times the $i$th row of $Q_{\k, \theta}$, then $s_j$ times the $j$th column of that intermediate matrix, given a choice of input state $\ket{s}$. Matrix $Q_{\k, \theta}$ is obtained by multiplying the parametrized unitary $U_\theta$ that represents the ansatz with the operator $W^{\k} = \text{diag} \left( (-1)^{k_1}, \cdots, (-1)^{k_m} \right)$.
\end{tcolorbox}
\caption{Estimation of the MMD for a linear optical QCBM.}
\label{fig:mmd_procedure}
\end{figure*}
For a given set of parameters $\theta$, the time complexity of estimating the MMD loss thus scales as $O(n^2/\epsilon_k^2 \epsilon_z^2$) if we take  $O(1/\epsilon_k^2)$ samples for set $\mathcal{A}$ and $O(1/\epsilon_z^2)$ samples for set $\mathcal{Z}$, resulting in a precision of $O(\epsilon_z \epsilon_k)$. Additional complexity polynomial in $m$ comes from constructing matrix $Q_{\k, \theta}$. %\textbf{Question: is this last statement correct?} 

Regarding precision, we note that the MMD as a classical ML metric is already subject to a limited precision due to the finite number of samples used in its estimation. In our case, the following effects influence the precision: finite number of samples for sets $\mathcal{X}, \mathcal{K}, \mathcal{Z}$ (target dataset, MMD observable, Gurvits estimate), as well as some imprecision when the no-collision assumption is not strictly respected. In our numerical tests, we found that these effects did not prevent the MMD to serve its purpose as a loss function, in terms of convergence and numerical stability.

\subsection{Kernel choice}\label{sec:kernel_choice}
Recall that in the proof of Proposition \ref{proposition:main}, the Fourier transform of the Gaussian kernel appeared and that it matches the $\mathcal{P}_{\sigma}(\k)$ of expression (\ref{eq:prob_p_sigma}). By choosing a $\mathcal{P}(\k)$ that is the Fourier transform of another kernel function, we can obtain the MMD observable for that kernel, and the proof will follow in the same way, provided that the kernel is stationary, i.e. $K(\x,\y) = K(\x - \y) = K(\z)$. This was first proven in the context of IQP circuits by the authors of \cite{kurkin2025universalitykerneladaptivetrainingclassically}. 

\begin{proposition}\label{proposition:kernel}
The MMD loss function for distributions $p, q$ defined on $\Tilde{\mathcal{X}}_{m, n}$ can be computed in the no-collision regime with the linear optical observable 
\begin{equation}
O_{\text{MMD, LO}} = \sum_{\k} \mathcal{P}(\k)\Tilde{W}^{\k} \otimes \Tilde{W}^{\k},
\nonumber
\end{equation}
where $\mathcal{P}(\k)$ is the Fourier transform of a kernel function $K(\x, \y)$ that is stationary and bounded.
\end{proposition}
Different kernels can be better suited for different types of data, and the kernel can be selected according to the dataset, or as part of a hyperparameter optimization. In this work, we consider two polynomial kernels (standard and parity polynomial), a weighted Gaussian kernel, as well as kernels directly defined by their Walsh or Fourier coefficients: one which we call Walsh low-order kernel, and a data-biased alternative. We define and detail those in Appendix \ref{app:kernels}.

\section{Ans{\"a}tze and initialization strategies in linear optical circuits}

\subsection{Interferometer choice}\label{sec:ansatze} 
Recall that our model, as a photon-native algorithm, directly exploits linear optical elements like beam-splitters and phase shifters instead of the usual hardware-agnostic gates from quantum computing. As explained in Section \ref{sec:generative_model}, any $m \times m$ unitary $U$ that describes the evolution of a single photon over $m$ modes can be decomposed as a product of Mach-Zender elements defined over two modes. This was first proven by Reck et al. in \cite{Reck_1994} with a triangular decomposition. The rectangular decomposition often used in hardware was proposed later by Clements et al. in \cite{Clements_2016}.

The rectangular and triangular meshes are not the only possible universal interferometer architectures as has been demonstrated in recent works \cite{Hamerly_2022, Fldzhyan:20, Basani:23, Crisan:25}. For instance, the 3-MZI architecture resembles the Clements architecture but adds a third 50:50 beam splitter to each standard MZI, enclosing the external phase shift. There is also the butterfly mesh, based on the sine-cosine decomposition of unitaries, that results in a self-similar pattern of MZIs. Though both meshes are universal, they induce markedly different sampling densities on $U(N)$ when compared to Reck and Clements schemes: drawing phase shifts from the same region of parameter space during initialization can lead to strong over- and under- representation of different unitaries, depending on the mesh geometry \cite{Hamerly_2024}.

In addition, we can take another approach and consider a parametrization of an $m \times m$ unitary $U$ which does not directly correspond to hardware, i.e., to an interferometer. This does not cause any issues, since with classical training the model can be trained over those parameters, and when the model is deployed, the trained unitary can be decomposed into any mesh which matches the available hardware. In our code, we implement a parametrization where a vector of real parameters is first mapped to a general complex matrix, on which a QR decomposition is then applied. This parametrization is compatible with the Haar measure, in the sense that it yields Haar-distributed unitaries when the parameters are sampled from an isotropic Gaussian distribution \cite{mezzadri2007generaterandommatricesclassical}.

\subsection{Input state}
The choice of input state $\ket{s}$ is, in a way, part of the ansatz choice. Often in linear optics, state $\ket{1\cdots10\cdots0}$ is selected as an input, but there is no reason why the $1$s cannot be placed in any input modes, as long as there is always a maximum of one photon per mode. We treat input state definition as a hyperparameter over which the model can be tuned.

The procedure presented in Figure \ref{fig:mmd_procedure} applies to simple Fock states as input states. However, classical simulation algorithms also exist \cite{lim2025} for Fock states prepared in superposition \cite{Loredo_2019}, i.e. for observable estimation in superposition boson sampling. Additionally, entangled input states of single photons can be produced experimentally, such as caterpillar states \cite{Huet_2025}, which are considered a crucial step in the creation of photonic cluster states \cite{Lindner_2009}. To our knowledge, there are currently no efficient classical algorithms for the estimation of observables in this last scenario, which is sometimes referred to as entangled boson sampling. It would be interesting to study the conditions under which classical simulation may be possible. Overall, the use of different input states will allow us to relax the constraint of fixed Hamming weight in the generated data, and would likely improve the expressivity of the model.

\subsection{Warm starts and initialization strategies}\label{sec:inits}
In both quantum and classical variational algorithms, it is well known that careful initialization can greatly improve the performance of an algorithm \cite{skorski2021revisiting, narkhede2021review, Zhang_2022}. We therefore propose several initialization strategies that depend on the choice of the interferometer mesh, building on approaches studied in the literature \cite{Grant_2019, Wang_2024}. In particular, we look for linear-optical counterparts to the strategies studied for IQP circuits in \cite{lerch2026iqpbornmachinesdatadependent}.

\paragraph{Near-identity initialization:}
This first strategy has been studied extensively in the literature and initializes the interferometer close to  identity, with small perturbations up to a chosen maximal value. For rectangular
interferometers and butterfly meshes, this corresponds to initializing the angles \(\phi\) and \(\theta\) close to \(0\). For 3-MZI meshes, the angles \(\phi\) and \(\theta\) are instead initialized close to \(\pi/2\). For the Haar-compatible parametrization, the unitary matrix is initialized close to the identity by adding a small perturbation to the identity matrix. This near-identity initialization is also meaningful in the multi-particle setting: if the interferometer unitary $U_\theta$ is close to the identity at the single-photon level, the induced transformation on the \(n\)-photon sector $\tilde{U}_\theta$ remains close to the identity, with a deviation scaling at most linearly in \(n\). We give a short proof of this statement in Appendix~\ref{app: close to identity}.

\paragraph{Fully random initialization:}
We also consider the strategy of initializing all circuit parameters randomly, where the distribution from which the parameters are drawn depends on the geometry of the chosen ansatz. For mesh-based ansätze, the phase shifts are sampled uniformly according to the natural parameter ranges of the corresponding architecture: for example, between $0$ and $2\pi$ for the phases in the Clements interferometer. For the Haar-compatible ansatz, the parameters define a complex matrix that is mapped to a unitary matrix through QR decomposition.

\paragraph{Data-dependent warm-start:}
In the IQP case, it is possible to choose circuit parameters so that single-qubit marginal distributions match those of the training dataset, and this data-dependent strategy proves to be very effective \cite{Recio-Armengol2025, lerch2026iqpbornmachinesdatadependent}. More precisely, given the target dataset
\(X_{\mathrm{train}}\in\{0,1\}^{N\times m}\), we define the empirical marginals
\[
p_j^{\mathrm{target}}
=
\frac{1}{N}\sum_{\ell=1}^N X_{\ell,j}.
\]
In the linear-optical case, we would have to match these values to the output single-mode marginals \(\widehat p_j(\theta)\). Contrary to the IQP case, it is not clear how to select specific circuit parameters to achieve this correspondence, but we can still obtain them heuristically -- thus making this approach a \emph{warm-start} rather than an initialization strategy.

We detail our method in Appendix \ref{app:data-dependent-init}. In a nutshell, single-mode marginals \(\widehat p_j(\theta)\) can be computed efficiently for given parameters $\theta$ from the first \(n\) columns of the interferometer using the marginal formula highlighted in \cite{Seron_2024} in the context of binned boson-sampling probabilities. As expected, this formula relies on permanents, and in the rank-one case considered here, the permanent can be rewritten using elementary symmetric polynomials \cite{Minc_1984}, which can be computed efficiently using dynamic programming. This creates a pipeline where the warm-start parameters can be obtained by heuristically minimizing
\[
\mathcal L_{\mathrm{marg}}(\theta)
=
\frac12
\sum_{j=1}^m
\left(
\widehat p_j(\theta)-p_j^{\mathrm{target}}
\right)^2.
\]
The resulting values are subsequently used as the initial point for the training of the model.

\paragraph{Unbiased initialization:}
In the case of IQP circuits, an unbiased initialization is defined which sets the model output distribution to a uniform distribution over all possible bit strings, and it yields the best results among data-agnostic strategies \cite{lerch2026iqpbornmachinesdatadependent}. In linear optics, finding an analogue to this initialization is again not straightforward. Finding exact uniformity from a computational point of view is out of reach due to the complexity of computing permanents, which we recall is \(\#P\)-hard in general \cite{aaronson2011}.
Moreover, Aaronson and Arkhipov showed that boson-sampling distributions associated with Haar-random interferometers are typically far from the uniform distribution, even when restricted to collision-free outcomes \cite{aaronson2014}. Thus, a single Haar-random interferometer should not be expected to produce an exactly uniform distribution.

At this stage, we can thus only propose a weaker notion of unbiasedness, which is perhaps closer in nature to random initializations: uniformity in expectation over the initialization. The proof is given in Appendix~\ref{app:unbiased-init}. More precisely, if the interferometer is drawn from the Haar measure, then no collision-free output configuration of fixed Hamming weight is preferred on average. For
\[
\widetilde{\mathcal X}_{m,n}
=
\{\x\in\{0,1\}^m : |\x|=n\},
\]
we have
\[
\mathbb{E}_{U\sim \mathrm{Haar}}[q_U(x)]
=
\frac{1}{\binom{m}{n}},
\qquad
\x\in\widetilde{\mathcal X}_{m,n}.
\]
For the Haar-compatible ansatz, in practice, the difference with the fully random initialization is that the parameters are drawn from a normal distribution instead of a uniform distribution. For the other ans{\"a}tze, the unitary is first obtained through this procedure, then decomposed into interferometer phases to obtain the initial parameters.

\textbf{Unbiased initialization for block structure:}
Interestingly, for some special cases, we can define a proper notion of unbiasedness in initializations. For example, if a dataset exhibits a certain block structure, we can find an initial set of parameters with the unbiasedness property. 

Suppose that the modes are divided into $\ell$ disjoint blocks, and that each data sample contains exactly one photon per block. In the simplest pairwise case, each block has size two and the data have the form: 
\begin{equation*}
|\x\rangle
=
|x_1\rangle\otimes\cdots\otimes |x_\ell\rangle
\otimes |0\cdots 0\rangle,
\quad
x_i\in\{10,01\}.
\end{equation*}
This can correspond for instance to pairwise-choice or one-hot-encoded data, where in each block, the photon encodes one of two possible choices, and extra empty modes may be added to respect the no-collision assumption. 

The idea of the initialization is to start from the fixed input state $|s_{\mathrm{in}}\rangle = |10\rangle^{\otimes \ell}\otimes |0\cdots 0\rangle,$  and to define an ansatz which applies a Hadamard transformation on each pair of modes. This provides a distribution which is uniform over the $2^\ell$ configurations with exactly one photon in each pair: $\{10,01\}^{\otimes \ell}\otimes |0\cdots 0\rangle.$ In the more general case where each block has size $d$, one can start with a state where there is one photon in each block, and apply a discrete Fourier transform $F_d$.

Naturally, this initialization in not uniform over the full collision-free space \(\widetilde{\mathcal X}_{m,n}\), but it is uniform over the structured support defined by the pairwise-choice encoding. So, it only makes sense to use this initialization if the dataset itself has the right structure, and not for a general dataset of fixed Hamming-weight bitstrings.

\paragraph{Bandwidth warm-start:}
Finally, for the case of the Gaussian kernel, we propose an additional warm-start method based on the bandwidth \(\sigma\). We train the model in several stages, each corresponding to a different value of \(\sigma\). The bandwidth determines which orders of correlations contribute most strongly to the MMD: a larger \(\sigma\) emphasizes lower-order operators, whereas a smaller \(\sigma\) gives more weight to higher-order operators. In our simulations, we observed that this strategy can improve the final value of the loss function, although the differences often remain below one order of magnitude.

\section{Datasets and applications}
Since our linear optical setup generates fixed Hamming-weight bitstrings, it is natural to aim at modelling fixed Hamming-weight bitstrings, instead of post-processing the circuit's output samples or mapping them to another set. In other words, we make use of the boson sampling inductive bias. This means that data points from the target datasets should be vectors $\x = (x_1, \dots x_m)$ where $\sum_i {x_i} = n$ and $x_i \in \{0, 1\}$. We propose datasets that match this format, and publish code to prepare or generate them.

\subsection{Boson sampling datasets}
As a first dataset to test the validity of our method, we generate boson sampling datasets using the Perceval \cite{Heurtel_2023} package. In particular, we make use of the CliffordClifford2017 backend based on \cite{CliffordClifford}, which is appropriate for sampling simulation. The software implementation allows us to go up to a maximum of $m = 256$ modes, and we select $n = 16$ photons. Eventually, generating this data classically will not be possible as $m$ and $n$ scale, but we could imagine a learning problem where the data comes from samples obtained on hardware.

\subsection{User preference datasets}
A type of problem that often comes up in data science concerns user preferences: for example, how does a user rank the films they have seen on a streaming platform, or the objects they have bought on an online marketplace? This data may be used in recommendation systems, or for grouping users in clusters and providing them with different client support or marketing strategies. To match the fixed Hamming-weight bitstrings requirement, we rephrase the question slightly and consider a user's $n$ favourite elements among a list of $m$ items. 

It turns out that a quite specific dataset exists for $m = 100$ and $n = 10$ within PrefLib: a library for preferences \cite{Mattei2013}, which provides 5000 samples of the 10 favourite sushi among a list of 100 sushi options. While this is not exactly an industry-ready use case, we use the opportunity to test our model on a real user preference dataset that remains of reasonable size.

Additionally, we prepared datasets for movie preference for different $m$ and $n$ based on data from the Netflix prize \cite{bennett2007netflix}, which was used historically for benchmarking tasks.

\subsection{Bioinformatics datasets}
Bioinformatics is a field where data often naturally appears in bitstring format, for example to encode DNA, or study the effect of compounds on genes \cite{MOECKEL20242289, Gao2024}. The Connectivity Map (CMap) is a database created by the Broad Institute \cite{Lamb2006, Subramanian2017NextGenerationConnectivity} that catalogs cellular signatures from perturbations with genetic and pharmacologic perturbagens, i.e. stores a description of how gene expression is affected by such compounds. We create datasets by starting from a raw file of compound responses from the CMap platform \cite{CLUE_website}, selecting a so-called gene universe to work with vectors of length $m$, then assigning $1$s to the $n$ most affected genes, i.e. those that exhibit the largest differential expression for a given compound, while the rest are assigned $0$s.

\subsection{Toy datasets for block structure}
In order to test the special case of the unbiased initialization for block structure presented in Section \ref{sec:inits}, we created two toy datasets. We call them toy datasets since the data requires some artificial manipulation to reach the right structure -- in particular with the padding of extra empty modes to respect the no-collision assumption.

For the first one, we start from an alternative version of the sushi dataset also available on PrefLib \cite{Mattei2013} that indicates a strict ranking between a subset of 10 sushi, we select pairs of sushi and encode the user preference for each pair so that each datapoint has a pairwise-choice structure, and we complete with padding. This toy dataset corresponds to the initialization with the Hadamard transformations applied on each pair of modes.

For the second one, we use a dataset from the UCI Machine Learning repository \cite{mushroom_73} describing mushrooms according to their physical characteristics. The features are categorical, so we use one-hot encoding followed by padding to recover the right data structure. This time, each block is not limited to two modes, and in fact, they each have different sizes. This data is well suited to the initialization where a Fourier transform $F_{d_b}$ is applied on each block of size $d_b$.

\begin{figure*}
	\centering
	\begin{subfigure}{0.4\linewidth}
		\includegraphics[width=\linewidth]{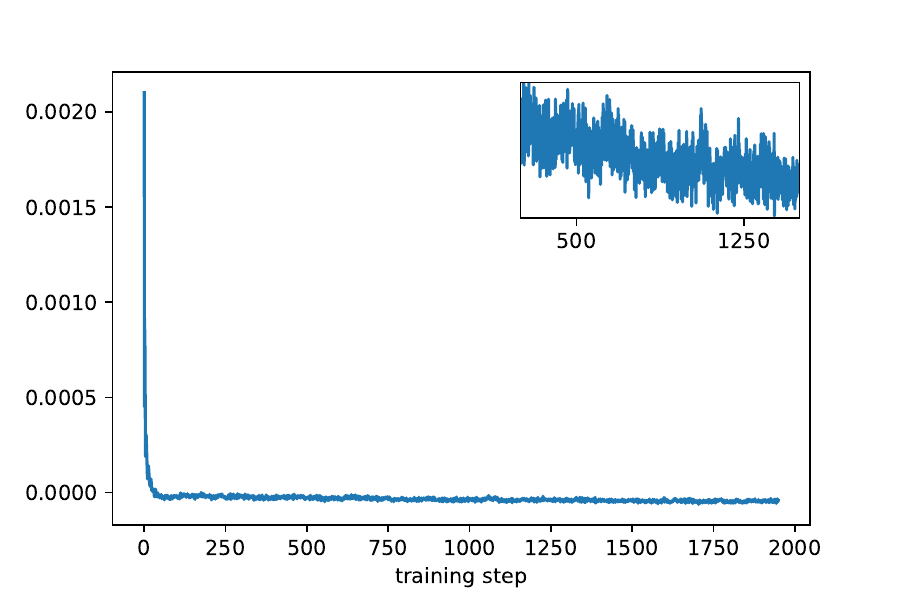}
		\caption{Boson sampling dataset, $m = 256$, $n = 16$, $\sigma = 3.0$, $|\mathcal{K}| = 2000$, $|\mathcal{Z}| = 2000$.}
		\label{fig:loss_evolution_A}
	\end{subfigure}
	\begin{subfigure}{0.4\linewidth}
		\includegraphics[width=\linewidth]{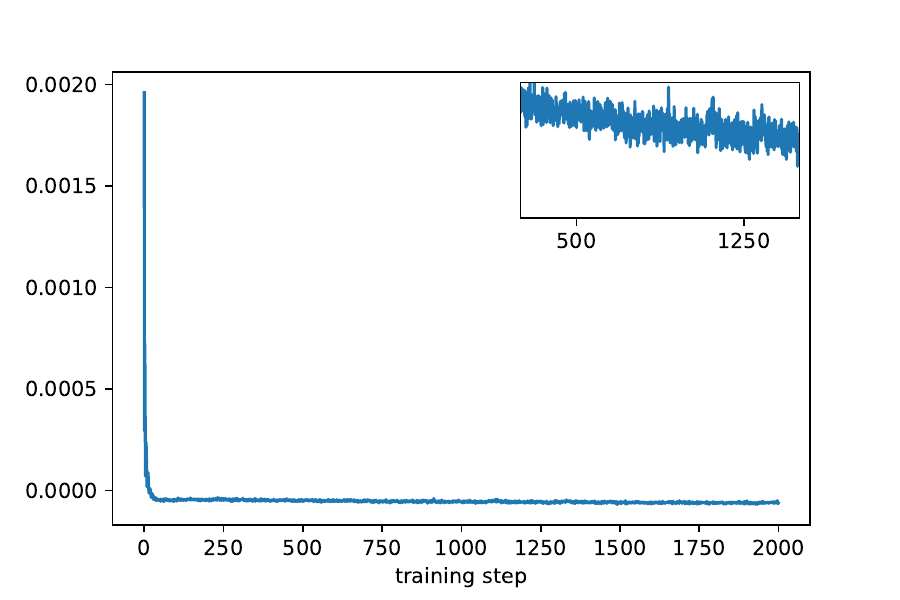}
		\caption{Boson sampling dataset, $m = 256$, $n = 16$, $\sigma = 3.0$, $|\mathcal{K}| = 5000$, $|\mathcal{Z}| = 5000$.}
		\label{fig:loss_evolution_B}
	\end{subfigure}
	\caption{Subfigures (a) and (b) compare the MMD loss curve for different sizes of sets $\mathcal{K}$ and $\mathcal{Z}$ from the estimation procedure: when the sets are smaller, the loss has slightly larger variations.}
	\label{fig:loss_evolution}
\end{figure*}

\section{Numerical experiments}\label{sec:tests}

\begin{figure*}
	\centering
    \begin{subfigure}{0.45\linewidth}
	    \includegraphics[width=\linewidth]{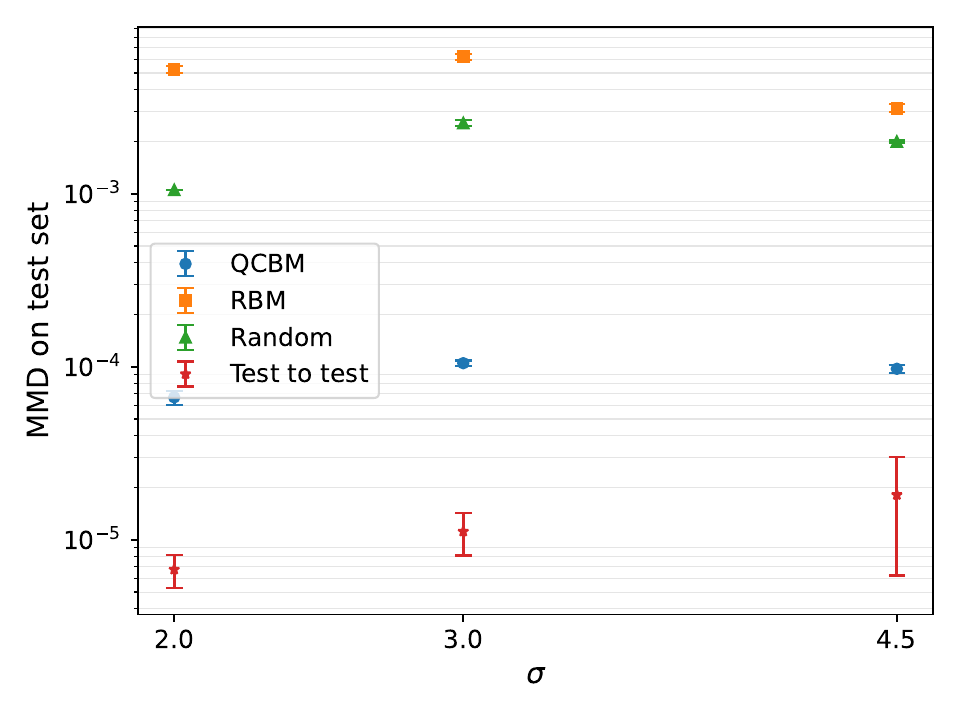}
            \captionsetup{width=0.9\linewidth}
	        \caption{Boson sampling dataset, $m = 256$, $n = 16$, near-identity initialization, Haar-compatible ansatz, comparing  different bandwidths for the Gaussian kernel.}
	    \label{fig:numerical_results_B}
    \end{subfigure}
    \begin{subfigure}{0.45\linewidth}
        \includegraphics[width=\linewidth]{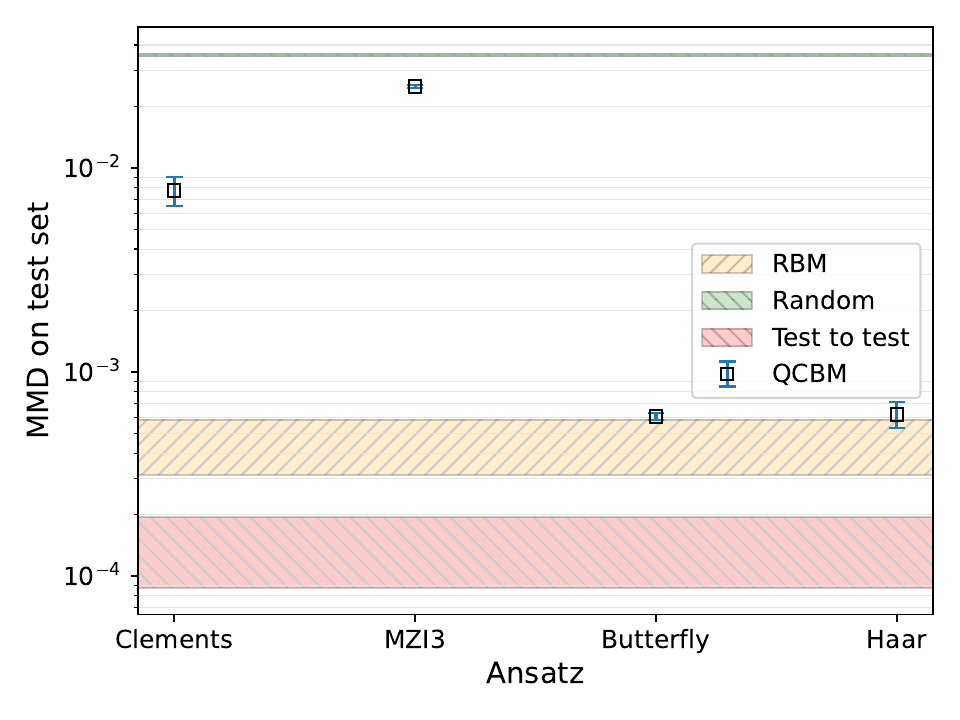}
	        \captionsetup{width=0.9\linewidth}
            \caption{User preference (sushi) dataset, $m = 100$, $n = 10$, Gaussian kernel for $\sigma = 3.0$, near-identity initialization, comparing different ans{\"a}tze.}
	    \label{fig:numerical_results_A}
    \end{subfigure}
         \begin{subfigure}{0.45\linewidth}
	        \includegraphics[width=\linewidth]{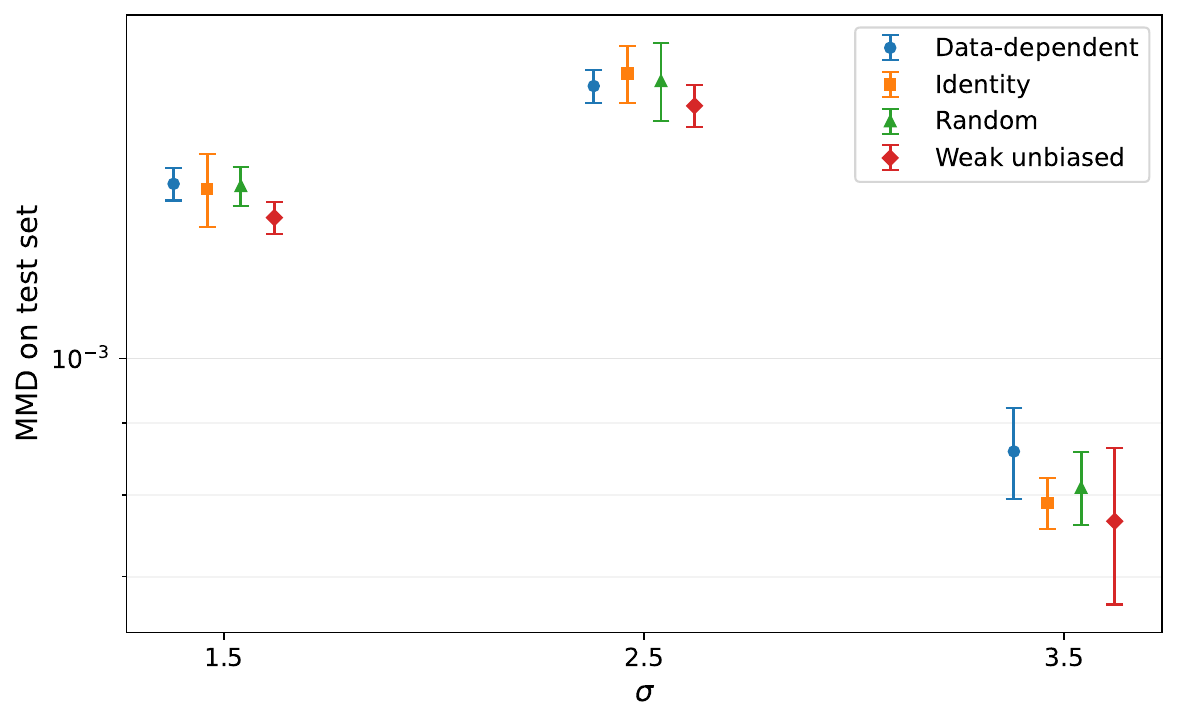}
	        \captionsetup{width=0.9\linewidth}
            \caption{User preference (movies) dataset, $m = 100$, $n = 10$, for the Haar-compatible ansatz, comparing different initialization strategies and bandwidths for the Gaussian kernel.}
	        \label{fig:numerical_results_C}
         \end{subfigure}
                  \begin{subfigure}{0.45\linewidth}
	        \includegraphics[width=\linewidth]{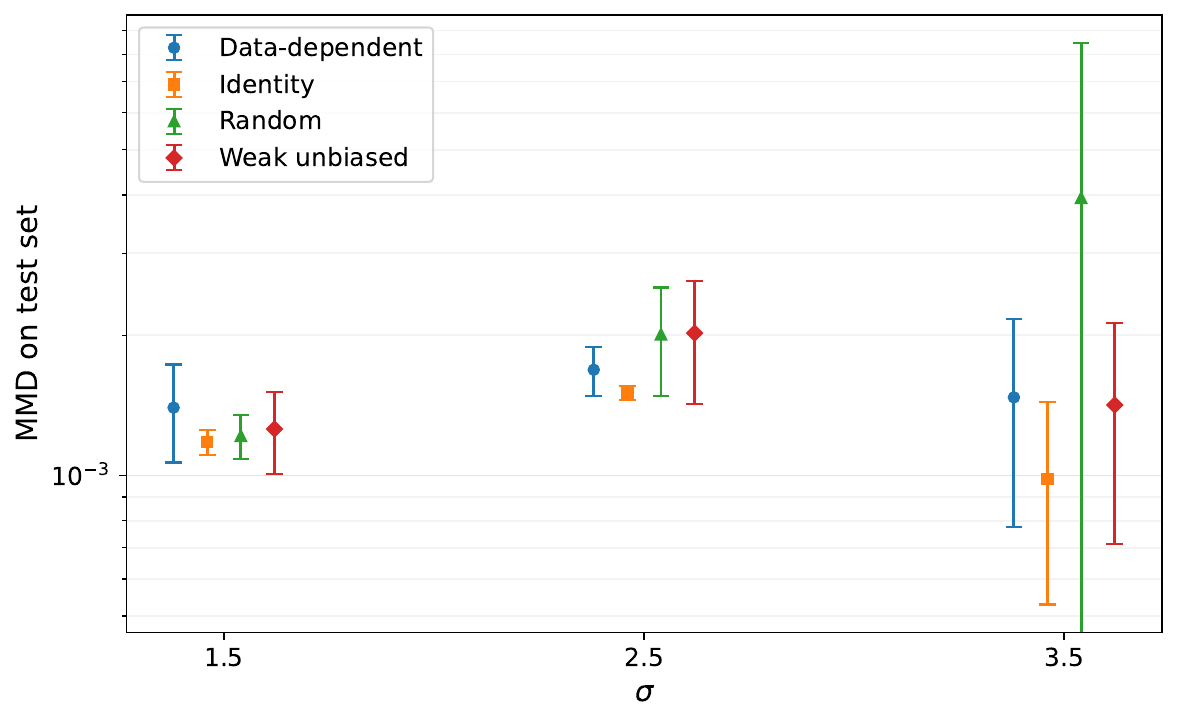}
	        \captionsetup{width=0.9\linewidth}
            \caption{User preference (movies) dataset, $m = 100$, $n = 10$, for the butterfly ansatz, comparing different initialization strategies and bandwidths for the Gaussian kernel.}
	        \label{fig:numerical_results_D}
         \end{subfigure}
        \begin{subfigure}{0.45\linewidth}
	        \includegraphics[width=\linewidth]{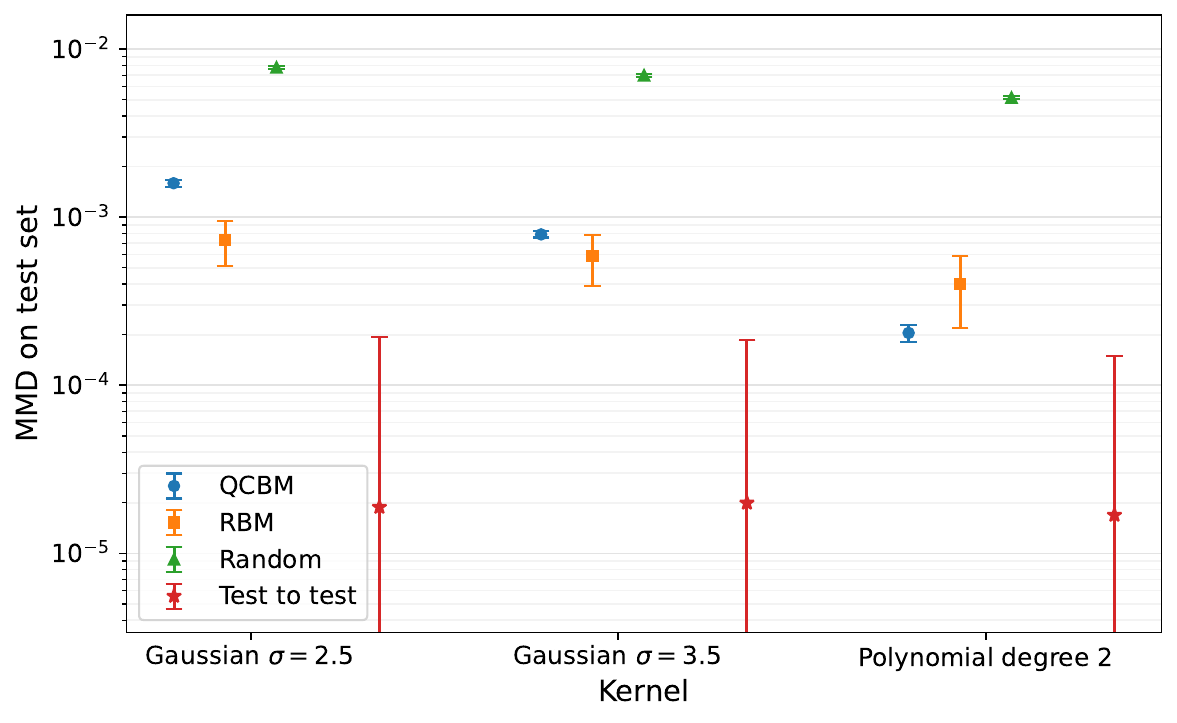}
	        \captionsetup{width=0.9\linewidth}
            \caption{User preference (movies) dataset, $m = 100$, $n = 10$, Haar-compatible ansatz, comparing Gaussian kernels to a polynomial kernel.}
	        \label{fig:numerical_results_E}
         \end{subfigure}
         \begin{subfigure}{0.45\linewidth}
	        \includegraphics[width=\linewidth]{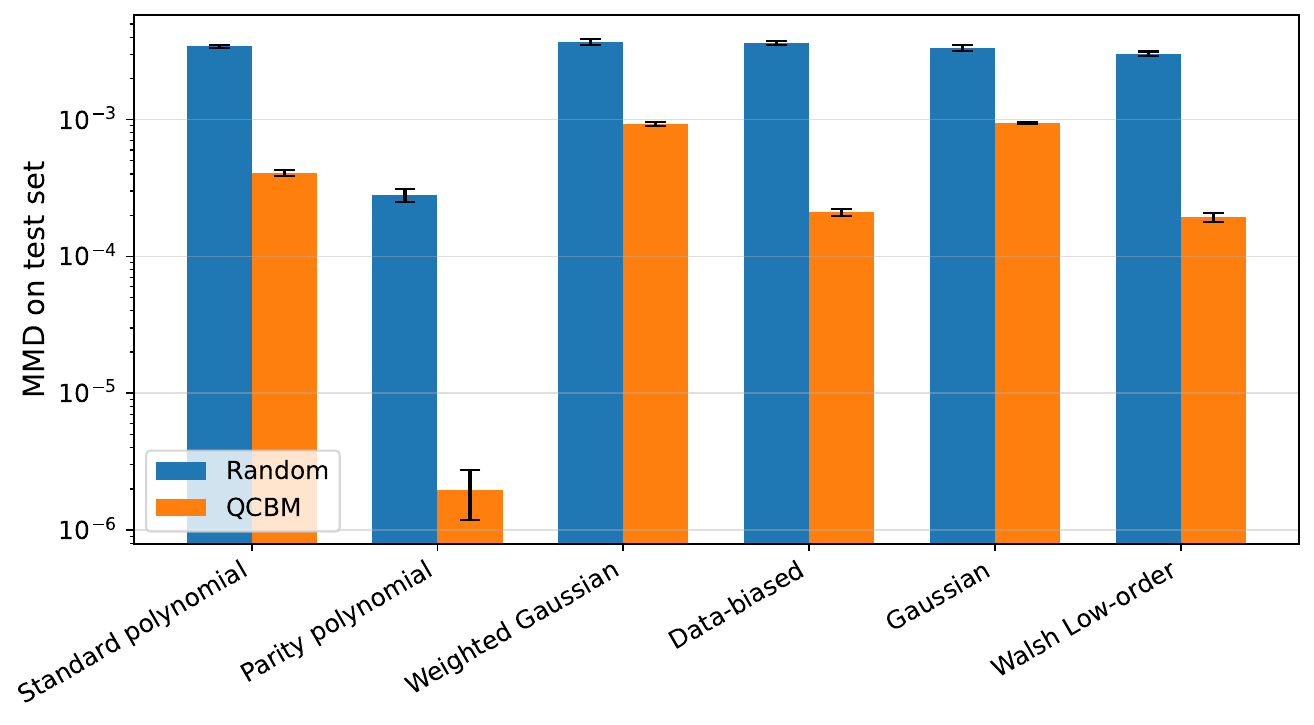}
	        \captionsetup{width=0.9\linewidth}
            \caption{Bioinformatics dataset, $m = 100$, $n = 10$, Haar-compatible ansatz, comparing different kernel choices with the random benchmark.}
	        \label{fig:numerical_results_F}
         \end{subfigure}
	\caption{Comparison of MMD values on test sets for the quantum model and the classical benchmarks, over various datasets, configurations and hyperparameter settings.  Subfigure (a) corresponds to the boson sampling dataset, where the quantum model performed best. Subfigure (b) shows that the butterfly and Haar-compatible ans{\"a}tze performed better, a trend we observed over most configurations and datasets. Subfigures (c) and (d) compare initialization strategies for the butterfly and the Haar-compatible ans{\"a}tze, and show that the butterfly interferometer seems more sensitive to initialization choice. Subfigures (e) and (f) illustrate the effect of kernel choice: overall, we observe that the polynomial kernel performs better. Error bars are obtained from repeating the simulations five times each, and in the case of the test-to-test MMD, by shuffling the test set.}
	\label{fig:numerical_results}
\end{figure*}

\begin{figure*}
	\centering
         \begin{subfigure}{0.45\linewidth}
	        \includegraphics[width=\linewidth]{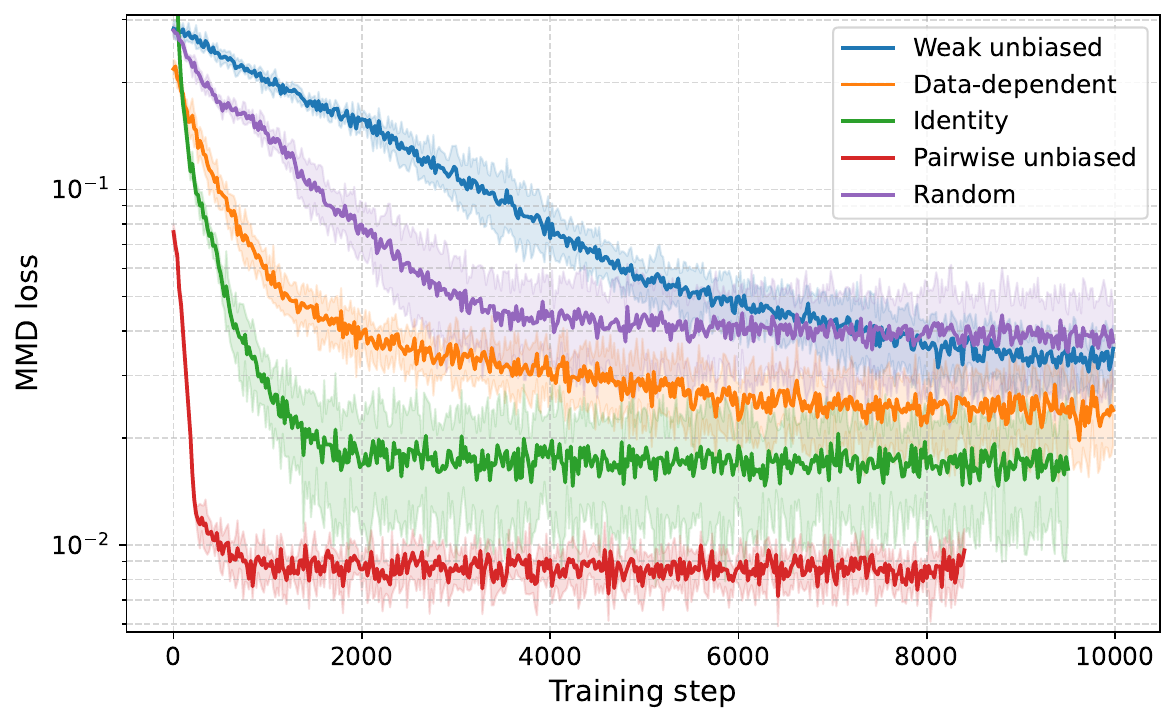}
	        \caption{Pairwise sushi dataset.}
	        \label{fig:init_A}
         \end{subfigure}
                  \begin{subfigure}{0.45\linewidth}
	        \includegraphics[width=\linewidth]{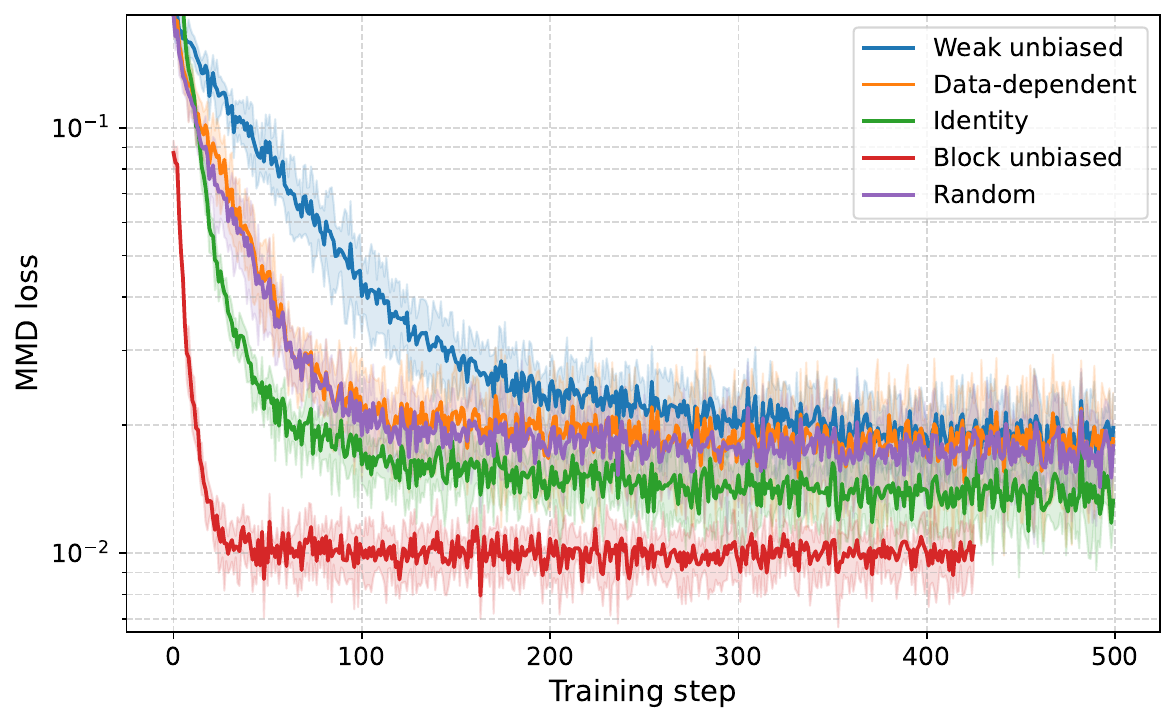}
	        \caption{Mushroom dataset.}
	        \label{fig:init_B}
         \end{subfigure}
	\caption{Exploring initialization strategies for the special case of data with block structure: loss curves for the sushi and mushroom toy datasets. Gaussian kernel with $\sigma = 2.0$, averaged over five training instances. Training ends either when a pre-set maximal number of steps or convergence criterion is reached.}
	\label{fig:pairwise}
\end{figure*}

We implemented the procedure from Figure \ref{fig:mmd_procedure} using Python and JAX. We relied on JAX's ability to automatically differentiate native Python and Numpy functions to carry out the optimization of the model. We used the Adam optimizer for gradient descent. As an example, a 2000-step training of a model with $m = 100$ modes and $n = 10$ photons takes about $2$ hours to complete on a laptop, with $3$ to $4$ seconds per step.

In Figure \ref{fig:loss_evolution}, we illustrate that training behaves as expected and that the loss function decreases smoothly. We also show the effect of choosing different sizes for sets $\mathcal{K}$ and $\mathcal{Z}$: the oscillations are more marked when the precision is lower. 

We set up a hyperparameter tuning framework using Ray Tune \cite{liaw2018tuneresearchplatformdistributed}, in order to explore many configurations including different kernel choices, ansatz choices, and initialization strategies. We summarise our findings in Figure \ref{fig:numerical_results}. We compare our models against two classical benchmarks: a simple restricted Boltzmann machine (RBM) like the one used in \cite{Recio-Armengol2025}, and a uniform random model that generates fixed Hamming-weight bitstrings. We also include the value of the MMD between two halves of the test set, as a ground truth ``test to test'' baseline. 

In terms of overall performance in the sense of outperforming our classical benchmarks, we observed that the QCBM performed best for the boson sampling dataset, a trend which remained valid over all configurations we explored. We see in Figure \ref{fig:numerical_results_B} that the quantum model clearly outperforms the classical benchmarks and reaches values much closer to the test-to-test ground truth MMD, for different values of the Gaussian kernel bandwidth. This could be expected, since both the model and the generator that produced the data are boson samplers, so the model naturally has the right inductive bias for such a dataset. 

In a sense, this confirms the guideline often repeated in QML, that quantum models are best suited for quantum data, or data obtained from quantum processes \cite{Huang_2021}. While the boson sampling dataset is not as clearly related to an application as the other datasets, this is an interesting result in its own right, and it indicates that results about quantum data obtained in the qubit picture do apply to the linear-optical or photon-native picture.

When exploring unitary parametrizations or ansatz choice, our simulations generally showed that the butterfly and Haar-compatible ans{\"a}tze yielded a better performance for our model. This seems logical given that the Clements and 3-MZI ans{\"a}tze tend to limit the propagation of photons over all the modes, especially for near-identity initializations, for example. We illustrate this in Figure \ref{fig:numerical_results_A} for the user preference (sushi) dataset. 

When exploring initializations, we did not observe large gaps in performance between strategies. Surprisingly, the data-dependent warm-start did not yield a much better performance than the others, perhaps because of imprecisions resulting from our heuristic method, in comparison with the exact initialization in the case of IQP circuits. We did observe an interplay with other hyperparameters: for instance, the butterfly ansatz with a larger bandwidth for the Gaussian kernel showed a high variability in performance, especially for the random initialization, as displayed in Figures \ref{fig:numerical_results_C} and \ref{fig:numerical_results_D}. This indicates that, when training over a new dataset or with a new kernel or ansatz choice, different initializations should be considered and tested.

For the special case of the unbiased initialization for block structure, we tested different intializations on our toy datasets which we show in Figure \ref{fig:pairwise}. We see that the unbiased initialization works very well in both cases. As a caveat, we recall that these toy datasets are slightly artificial. Also, this specific case may yield optimal ans{\"a}tze that display a lot of structure (closer to Fourier interferometers), which might make the model more vulnerable to classical simulation in the deployment phase. However, this still shows a nice example of a particularly efficient initialization strategy.

Then, when we explored the effect of kernel choice, we saw that the Gaussian kernel was interestingly  not always the most adequate, and that the polynomial kernel often offered better performance, as shown in Figures \ref{fig:numerical_results_E} and \ref{fig:numerical_results_F} for the user preference and bioinformatics datasets, respectively. In Figure \ref{fig:numerical_results_E}, we see that switching to the polynomial kernel allows the quantum model to slightly overtake the RBM benchmark. In Figure \ref{fig:numerical_results_F}, we highlight the comparison between the random model and the QCBM, since absolute values of the MMD should not be compared when changing kernel, but instead performance relative to benchmarks.

Still, we must note that our classical benchmarks are fairly simple models, and that there remains a gap between all models and the ground truth MMD. It is true that linear optical models are not universal models, as pointed out in the independent work of \cite{kurkin2026universalityclassicallytrainablequantumdeployed} and they can lack in expressivity, even when considering fixed-Hamming-weight bitstring datasets, as the space of all unitaries cannot be reached when mapping a one-photon unitary $U_\theta$ to an $n$-photon operator $\tilde{U}_\theta$. In the independent work of \cite{kurkin2026universalityclassicallytrainablequantumdeployed}, the authors propose constructions that may increase the expressivity of linear optical models while retaining hardness of sampling and efficiency of classical training, at the cost of increasing the resources in terms of number of modes $m$. It would be very interesting to test these constructions numerically in future work.

\section{Discussion}
In this work, we present and implement a method to train photon-native quantum generative models efficiently on classical hardware. Importantly, once the model is trained, deploying the model, i.e. sampling from it, coincides with the problem of boson sampling, a task conjectured to be classically hard and thus requiring quantum hardware in general. Of course, the parameters obtained during training may define unitaries with a structure too specific to prove any complexity-theoretic statement, as was rightfully highlighted recently in \cite{herbst2025limitsquantumgenerativemodels}; or their structure may even allow for classical simulability of sampling. Nevertheless, this approach allows us to propose a photonic QML algorithm which is both scalable and has a decent chance to escape dequantization.

With this work, we add a new string to the bow of "train on classical, deploy on quantum" methods. In the short timespan since the original proposal of \cite{Recio-Armengol2025} came out, many related studies have emerged \cite{kurkin2025universalitykerneladaptivetrainingclassically, herrerogonzalez2025bornultimatumconditionsclassical, ballogimbernat2025shallowiqpcircuitgraph, bako2025fermionicbornmachinesclassical}, some of which we have already had the opportunity to highlight in the text. Particularly relevant to our case are the independent works of \cite{kurkin2026universalityclassicallytrainablequantumdeployed, kolarovszki2026generativemodelinggaussianboson} on photonic generative models based on Fock-state boson sampling and Gaussian boson sampling respectively -- in the latter case which corresponds to continuous variable photonics, the authors reach the remarkable amount of $10^6$ trainable parameters in their numerical simulations. And from the experimental perspective, a first demonstration for shallow IQP circuits was carried out by \cite{ballogimbernat2025shallowiqpcircuitgraph} with model deployment  on superconducting hardware.

These methods are part of a larger trend of using quantum circuits only where they are strictly necessary within a classical pipeline, as for the "quantum-enhanced classical simulation" paradigm \cite{Cerezo_2025, zimboras2025mythsquantumcomputationfault, Huang_2022} -- including recent photonics results for classical shadow protocols in \cite{thomas2025sheddinglightclassicalshadows}. This trend aligns well with the impressive capabilities of classical computers for certain algorithms: for example, training the circuit classically allows us to use tools like automatic differentiation.

The "train on classical, deploy on quantum" paradigm is another reason to believe that generative learning is a promising area to look for utility of boson sampling. While our work is only a step in that direction, we believe that the datasets and use cases we propose are a good a starting point, and we believe there is potential in looking for more applications in this direction.

Many open questions remain: for instance, it would be interesting to extend our work to the collision case, as well as to different types of input states, as discussed in the text. Additionally, we can certainly explore cases with much larger number of photons and modes if we exploit more significant classical computing resources and high-performance computing.

Deploying our models on actual devices appears as a natural next step -- in fact, this may be a strength of using linear optical models, as the corresponding circuits are relatively easy to deploy experimentally. Indeed, all required optical components are readily available and interferometers are steadily increasing in size, although the question of noise, in particular photon loss, remains a challenge. Then, it will be crucial to study our models under noisy conditions \cite{Aaronson_2016, Garc_a_Patr_n_2019}, and those will certainly imply higher thresholds in terms of the number of photons and modes required to reach an advantage regime for model deployment -- in which case, studying the energetic advantage perspective as was done in \cite{soret2026quantumenergeticadvantagecomputational} will be of significant interest.

\section*{Acknowledgements}
The authors would like to thank Joseph Bowles and Pierre-Emmanuel Emeriau for fruitful discussions. This work has been co-funded by
the QuantERA project ResourceQ under the grant
agreement ANR-24-QUA2-007-003, and by the EIC
Pathfinder program project QUONDENSATE under
the grant agreement 101130384.

\section*{Code availability}
Our code is available within Quandela's MerLin platform, at \url{https://github.com/merlinquantum/generative_classical_training}.

\section*{Author Contributions}
F.G. and C.F. carried out the development of the code pipeline, contributed to the theoretical framework and to the manuscript writing. R.M. and B.V. contributed to the theoretical framework and manuscript writing. S.M. reviewed the manuscript. A.S. led the project, and contributed to the code pipeline, theoretical framework and manuscript writing. Large language model tools were used to assist with the development of our code. All scientific content, calculations, simulations, conclusions, and final text were reviewed and approved by the authors, who take full responsibility for the work.

%\begin{figure}
%\centering
%\includegraphics[width=0.3\textwidth]{frog.jpg}
%\caption{\label{fig:frog}This frog was uploaded via the file-tree menu.}
%\end{figure}

\bibliographystyle{unsrt}
\bibliography{sample}

\appendix
\onecolumngrid

\section{Kernel choice}\label{app:kernels}
First, we consider two polynomial kernels. The first one is the \textbf{standard polynomial kernel}:
\[
K_{\mathrm{poly,1}}(\x,\y)
=
\left(
\frac{c+\x\cdot \y}{c+n}
\right)^d,
\qquad |\x|=|\y|=n.
\]
On the subspace of fixed Hamming-weight bitsrings, the kernel is stationary, since:
\[
\x\cdot \y
=
n-\frac{|\x\oplus \y|}{2}.
\]
However, its Walsh (or Fourier) coefficients are not always non-negative, in which case they cannot properly define a probability distribution $\mathcal{P}(\k)$ for the MMD observable. Therefore, this
kernel can only be used directly in our estimator under additional conditions on the parameters. We find that a sufficient condition is to choose $c$ such that:
\[
c+n-\frac{m}{4}\geq 0.
\]
The second polynomial kernel is a \textbf{polynomial kernel on parity features}. We first define:
\[
\Phi(x)=\frac{1}{\sqrt{m}}((-1)^{x_1},...,(-1)^{x_m}),
\]
and then
\[
K_{poly,2}(\x,\y)=(\frac{c+\Phi(\x).\Phi(\y)}{c+1})^d.
\]
Equivalently, the kernel can be written as:
\[
K_{poly,2}(\x,\y)=(\frac{c+\frac{1}{m}\sum_{i=1}^m(-1)^{x_i \oplus x_j}}{c+1})^d. 
\]
This kernel is naturally adapted to the Walsh observables, i.e. to the operators  $W^{\k}$ of equation (\ref{eq:z_mbym_matrix}) of the main text, since it is built directly from parity features. Its Walsh coefficients are non-negative, so they define a  valid sampling distribution $\mathcal{P}(\k)$.

As an alternative approach, instead of starting from a usual kernel formula in data space, we directly choose a probability distribution over the $W^{\k}$ observables. We call those \textbf{low-order Walsh kernels}. For a maximum order $r$, we sample $\k$ with \(1\leq |\k|\leq r\). More precisely,
\[
P_{\mathrm{low}}(k)
=
\begin{cases}
\dfrac{1}{r}\dfrac{1}{\binom{m}{|\k|}},
& 1\leq |\k|\leq r,\\[1.2ex]
0,
& \text{otherwise}.
\end{cases}
\]

Finally, we introduce two data-adaptive kernels, where the kernel is defined based on the target dataset. The first one is a \textbf{weighted Gaussian kernel}, an alternative to the Gaussian kernel: 
\[
K_{\mathrm{wG}}(\x,\y)
=
\exp\left(
-\frac{1}{2\sigma^2}
\sum_{i=1}^m w_i (x_i-y_i)^2
\right).
\]
Here, we add weights $w_i$ to the usual Gaussian kernel, chosen according to the empirical marginal occupations of the modes in the training dataset $\mu_i=\frac{1}{N}\sum_{l=1}^N x_i^{l}$. Thus:
\begin{equation}
w_i=\frac{\mu_i+\epsilon}{\sum_{j=1}^m (\mu_i+\epsilon)}.
    \label{weights}
\end{equation}

The second data-adaptive variant is a \textbf{data-biased low-order Walsh kernel}. It keeps the same low-order structure as the low-order Walsh kernel, but the support of $\k$ is not chosen uniformly over all modes. Instead, modes are sampled according to equation (\ref{weights}). Thus, Walsh observables $W^{\k}$ involving frequently occupied modes are sampled over often. This gives a valid probability distribution $P_{biased}(\k)$, and therefore a valid positive semidefinite stationary kernel.

Recall that all the above kernels can be inserted into the MMD estimator from Figure \ref{fig:mmd_procedure}, keeping the estimator identical but changing only the sampling distribution $\mathcal{P}(\k)$. This allows us to easily compare how different choices of observables affect the training of the QCBM.

\section{Initialization strategies}
\subsection {Close to identity initialization - multiple photons }
\label{app: close to identity}
In this appendix, we show that a close-to-identity initialization at the single-particle level remains close to the identity when lifted to the \(n\)-photon sector.

Following the standard representation of passive linear optics, an
\(m\)-mode interferometer can be written as
\[
U(\boldsymbol{\epsilon})
=
\exp\left(
i\sum_j \epsilon_j h^j
\right),
\]
where the \(h^j\)'s are single-particle generators of the linear-optical
circuit and the coefficients
\[
\boldsymbol{\epsilon}=(\epsilon_j)_j
\]
are small initialization parameters. The induced transformation on the
\(n\)-photon sector is given by the bosonic representation
\(\phi_m^n(U)\). We refer to \cite{aaronson2011} for the formal
construction of this representation and of the induced generators.

For small initialization parameters, the Taylor expansion gives
\[
U(\boldsymbol{\epsilon})
=
I_m
+
i\sum_j \epsilon_j h^j
+
\mathcal O(\|\boldsymbol{\epsilon}\|^2).
\]
For a photon initially injected in mode \(i\), this implies
\[
a_i^\dagger
\longmapsto
a_i^\dagger
+
i\sum_{r=1}^m
\left(
\sum_j \epsilon_j (h^j)_{ri}
\right)
a_r^\dagger
+
\mathcal O(\|\boldsymbol{\epsilon}\|^2).
\]
Thus, at first order, the photon remains mostly in its initial mode, with
small amplitudes to be displaced to other modes.

Now consider a collision-free \(n\)-photon input
\[
|S\rangle
=
a_{i_1}^\dagger\cdots a_{i_n}^\dagger |0^m\rangle.
\]
The induced transformation acts by applying the same mode transformation to
each creation operator:
\[
\phi_m^n(U(\boldsymbol{\epsilon}))|S\rangle
=
\prod_{\ell=1}^n
\left(
\sum_{r=1}^m U(\boldsymbol{\epsilon})_{r i_\ell}a_r^\dagger
\right)
|0^m\rangle.
\]
Using the first-order expansion of \(U(\boldsymbol{\epsilon})\), we obtain
\[
\phi_m^n(U(\boldsymbol{\epsilon}))|S\rangle
=
|S\rangle
+
i\sum_{\ell=1}^n
a_{i_1}^\dagger\cdots
\left[
\sum_{r=1}^m
\left(
\sum_j \epsilon_j (h^j)_{r i_\ell}
\right)
a_r^\dagger
\right]
\cdots
a_{i_n}^\dagger |0^m\rangle
+
\mathcal O(n^2\|\boldsymbol{\epsilon}\|^2).
\]
The first-order correction contains \(n\) terms, corresponding to the choice
of which photon is weakly displaced. Hence it scales as
\[
\mathcal O(n\|\boldsymbol{\epsilon}\|).
\]
The second-order remainder comes from terms where two photons are perturbed,
giving the bound
\[
\mathcal O(n^2\|\boldsymbol{\epsilon}\|^2).
\]
Therefore, for each collision-free input configuration,
\[
\phi_m^n(U(\boldsymbol{\epsilon}))|S\rangle
=
|S\rangle
+
\mathcal O(n\|\boldsymbol{\epsilon}\|)
+
\mathcal O(n^2\|\boldsymbol{\epsilon}\|^2).
\]
In particular, if \(n\|\boldsymbol{\epsilon}\|\ll 1\), the close-to-identity
initialization remains close to the identity on the \(n\)-photon sector.

\subsection{Data-dependent warm start}
\label{app:data-dependent-init}

This appendix details the computation used in the data-dependent initialization.
The goal is to initialize the interferometer such that its single-mode marginals
match the empirical marginals of the training dataset.

Let
\[
X_{\mathrm{train}}\in\{0,1\}^{N\times m}
\]
be the training dataset, where each row is a collision-free configuration with
fixed Hamming weight \(n\). The empirical marginal occupation of mode \(j\) is
defined as
\[
p_j^{\mathrm{target}}
=
\frac{1}{N}
\sum_{\ell=1}^{N}
X_{\ell,j}.
\]
Thus \(p_j^{\mathrm{target}}\) represents the empirical probability that output
mode \(j\) is occupied in the training set.

We consider an input state with one photon in each of the first \(n\) modes.
Let \(U(\theta)\in\mathbb{C}^{m\times m}\) be the unitary implemented by the
interferometer. Since only the first \(n\) input modes are occupied, the relevant
quantities are the first \(n\) columns of \(U(\theta)\). For each output mode
\(j\), define
\[
S_{j,i}(\theta)
=
|U_{j,i}(\theta)|^2,
\qquad
i=1,\ldots,n.
\]

Following the binned boson-sampling marginal formula of \cite{Seron_2024}, the
probability that no photon is detected in output mode \(j\) can be written as
\[
P^{(j)}(0)
=
\operatorname{perm}
\left(
I_n - H_n^{(j)}
\right),
\]
where
\[
H_n^{(j)}
=
u_j u_j^\dagger,
\]
and \(u_j\) is the restriction of the \(j\)-th row of \(U(\theta)\) to the first
\(n\) columns:
\[
u_j
=
\left(
U_{j,1}(\theta),
\ldots,
U_{j,n}(\theta)
\right).
\]
Equivalently,
\[
H_{a,b}^{(j)}
=
U_{j,a}^{*}(\theta)U_{j,b}(\theta),
\qquad
a,b=1,\ldots,n.
\]

In the no-collision regime, the probability that mode \(j\) is occupied is
approximated by
\[
\widehat p_j(\theta)
\simeq
P^{(j)}(1)
\simeq
1-P^{(j)}(0).
\]
Hence,
\[
\widehat p_j(\theta)
=
1-
\operatorname{perm}
\left(
I_n-u_j u_j^\dagger
\right).
\]

The matrix \(u_j u_j^\dagger\) has rank one. Therefore, using standard identities
for permanents of rank-one perturbations of the identity \cite{Minc_1984}, we can
rewrite
\[
\operatorname{perm}
\left(
I_n-u_j u_j^\dagger
\right)
=
\sum_{k=0}^{n}
(-1)^k k!\,
e_k
\left(
|U_{j,1}(\theta)|^2,
\ldots,
|U_{j,n}(\theta)|^2
\right),
\]
where \(e_k\) denotes the elementary symmetric polynomial of degree \(k\). Since
\(e_0=1\), we obtain
\[
\widehat p_j(\theta)
=
\sum_{k=1}^{n}
(-1)^{k+1} k!\,
e_k
\left(
S_{j,1}(\theta),
\ldots,
S_{j,n}(\theta)
\right).
\]

The elementary symmetric polynomials are computed efficiently by dynamic
programming. For a vector
\[
s^{(j)}
=
\left(
S_{j,1},\ldots,S_{j,n}
\right),
\]
we initialize
\[
e_0=1,
\qquad
e_k=0
\quad
\text{for } k\geq 1.
\]
Then, for each entry \(s_i^{(j)}\), we update backwards:
\[
e_k
\leftarrow
e_k+s_i^{(j)}e_{k-1},
\qquad
k=i,i-1,\ldots,1.
\]
This computes all elementary symmetric polynomials
\(e_0,\ldots,e_n\) in \(O(n^2)\) operations for a fixed mode \(j\), and therefore
in \(O(mn^2)\) operations for all output modes.

The data-dependent initialization is then obtained by minimizing the auxiliary
loss
\[
\mathcal L_{\mathrm{marg}}(\theta)
=
\frac12
\sum_{j=1}^{m}
\left(
\widehat p_j(\theta)
-
p_j^{\mathrm{target}}
\right)^2.
\]
The optimized parameters are used as the initial parameters for the subsequent MMD training. Naturally, this procedure does not try to match the full target distribution at initialization. It only aims at matching first-order occupation statistics. The full training objective is still optimized afterwards using the MMD loss.

\subsection{Unbiasedness in expectation under Haar initialization}
\label{app:unbiased-init}

We justify here why drawing the interferometer from the Haar measure gives an
unbiased initialization in expectation over collision-free configurations.

Let
\[
\widetilde{\mathcal X}_{m,n}
=
\{\x\in\{0,1\}^m:\ |\x|=n\}
\]
be the set of collision-free configurations with \(n\) photons in \(m\) modes.
For a fixed input state, let \(q_U(\x)\) denote the probability of observing the
output configuration \(x\) after applying the interferometer \(U\).

Let \(\x,\x'\in\widetilde{\mathcal X}_{m,n}\). Since both configurations have the
same Hamming weight, there exists a permutation matrix \(P\) such that
\(P|\x\rangle=|\x'\rangle\). The Haar measure is invariant under multiplication by
a fixed unitary matrix, hence \(PU\) has the same distribution as \(U\). Therefore,
\[
\mathbb{E}_{U\sim\mathrm{Haar}}[q_U(\x)]
=
\mathbb{E}_{U\sim\mathrm{Haar}}[q_{PU}(\x)]
=
\mathbb{E}_{U\sim\mathrm{Haar}}[q_U(\x')].
\]
Thus all collision-free configurations have the same expected probability. Since
the probabilities sum to one over \(\widetilde{\mathcal X}_{m,n}\), we obtain
\[
\sum_{\x\in\widetilde{\mathcal X}_{m,n}}
\mathbb{E}_{U\sim\mathrm{Haar}}[q_U(\x)]
=
1.
\]
There are \(\binom{m}{n}\) such configurations, so
\[
\mathbb{E}_{U\sim\mathrm{Haar}}[q_U(\x)]
=
\frac{1}{\binom{m}{n}},
\qquad
\x\in\widetilde{\mathcal X}_{m,n}.
\]
This shows that Haar initialization does not favor any collision-free output
configuration on average, even though a single Haar-random interferometer is not
itself uniform \cite{aaronson2014}.

\end{document}